\documentclass[aip,jcp,preprint,showkeys,superscriptaddress,longbibliography,noeprint]{revtex4-2}
\usepackage[separate-uncertainty=true, multi-part-units=single]{siunitx}
\usepackage[inline,shortlabels]{enumitem}
\usepackage[T1]{fontenc}
\usepackage[utf8]{inputenc}
\usepackage[english]{babel}
\usepackage[a4paper, total={7in, 10in}]{geometry} 
\usepackage{microtype}
\usepackage{indentfirst}
\usepackage{lmodern}
\usepackage[intlimits]{amsmath}
\usepackage{amsthm, amssymb, amsfonts}
\usepackage{braket}
\usepackage{chngcntr}
\usepackage{csquotes}
\usepackage[pdfstartview={FitH}, colorlinks=true, linkcolor=blue, citecolor=green]{hyperref}
\newcommand{\N}{\mathbb{N}}
\newcommand{\R}{\mathbb{R}}
\newcommand{\MH}{\mathcal{H}}
\newcommand{\MO}{\mathcal{O}}
\newcommand{\MR}{\mathcal{R}}

\newcommand{\ee}{\mathrm{e}}
\newcommand{\diff}{\mathrm{d}}
\newcommand{\tand}{\quad\text{and}\quad}
\DeclareMathOperator{\dom}{dom}
\DeclareMathOperator{\Tr}{Tr}

\newtheorem{theorem}{Theorem}[section]
\newtheorem{lemma}[theorem]{Lemma}
\newtheorem{cor}[theorem]{Corollary}
\newtheorem{remark}[theorem]{Remark}
\theoremstyle{definition}
\newtheorem{definition}[theorem]{Definition}

\usepackage{graphicx,grffile,subfigure}
\graphicspath{{figures/}}
\newlength\figwidth
\usepackage{url,hyperref}
\usepackage[table,usenames,dvipsnames]{xcolor}
\hypersetup{colorlinks=true, linkcolor=BrickRed,
  urlcolor=blue!50!black, citecolor=blue!50!black}

\usepackage[capitalize]{cleveref}
\crefrangelabelformat{equation}{(#3#1#4)$-$(#5#2#6)}

\makeatletter
\def\paragraph{%
  \@startsection
    {paragraph}{4}{\parindent}{\z@}{-1.5em}%
    {\normalfont\normalsize\itshape}%
}%
\makeatother

\begin{document}

\title{{Two-sided Bogoliubov inequality to estimate finite size effects in quantum molecular simulations}}\thanks{This article belongs to the themed collection:\\Mathematical Physics and Numerical Simulation of Many-Particle Systems; V.Bach and L.Delle Site (eds.)}

\newcommand\FUBaffiliation{\affiliation{Freie Universität Berlin, Institute of Mathematics, Arnimallee 6, 14195 Berlin, Germany}}
\newcommand\BTUaffiliation{\affiliation{Brandenburgische Technische Universität Cottbus-Senftenberg, Institute of Mathematics, Konrad-Wachsmann-Allee 1, 03046 Cottbus, Germany}}
\author{Benedikt Reible}
\email{benedikt.reible@studserv.uni-leipzig.de}
\FUBaffiliation
\author{Carsten Hartmann}
\email{carsten.hartmann@b-tu.de}
\BTUaffiliation
\author{Luigi Delle Site}
\email{luigi.dellesite@fu-berlin.de}
\FUBaffiliation

\begin{abstract}
We generalise the two-sided Bogoliubov inequality for classical particles [L. Delle Site {\it et al.}, J.~Stat.~Mech.~Th.~Exp. {\bf 083201} (2017)] to systems of quantum particles. As in the classical set-up,  the inequality leads to upper and lower bounds for the free energy difference associated with the partitioning of a large system into smaller, independent  subsystems. From a thermodynamic modelling point of view, the free energy difference determines the finite size correction needed to consistently treat a small system as a representation of a large system. Applications of the bounds to quantify finite size effects are ubiquitous in physics, chemistry, material science, or biology, to name just a few; in particular it is relevant for molecular dynamics simulations in which a small portion of a system is usually taken as representative of the idealized large system. 
\end{abstract}

\maketitle

\section{Introduction}
Realistic physical systems are far too large for being treated at the level of single particles in full resolution, and thus a common approach is to consider small systems whose computation can be carried out at reasonable computational costs. The implicit assumption is that the corresponding statistical mechanics and thermodynamics represent the true physical situations within an acceptable degree of precision when the finite size effects are negligible in comparison with some reference quantity of interest.
While this situation occurs in many fields of physics, chemistry and material science, a field in which such an approximation is routinely used is molecular simulation \cite{frenkel,tuckermann}. Molecular simulation has made an enormous progress in the latest decades in successfully studying classical and quantum particle systems, but without the possibility of simulating a small system as a representative of an ideal infinite system, its power would have been modest due to the limitation of computational resources and the difficulties of data storage\cite{neuron}.

When replacing an infinite or very large system by a considerably smaller subsystem, the modelling error  can be large, and so can be the statistical error when averages over finitely many particles are considered\cite{cortes1,cortes2}. From this perspective, criteria that allow for a precise estimate of the finite size effects play a key role in the assessment of the quality of a simulation study. In a previous work \cite{jstatmecc}, we have derived two-sided Bogoliubov bounds for the interface free energy required for the separation of a classical infinite system into weakly interacting small subsystems that can be used as an error indicator of the model fidelity as discussed above. 
The upper and lower bounds of the interface free energy provide quantitative and computable error bounds to quantify the relevance of the size effects. It is moreover possible to derive tight variational versions of the bounds that can be the basis for systematic improvements of available approximate bounds. 

The aim of this paper is to generalise the bounds of Ref.~\citenum{jstatmecc} for classical particle systems to quantum systems. The classical Bogoliubov bounds rely on a change of measure of the underlying equilibrium probability measure and the non-negativity of the relative entropy between these probability measures; the difficulty here is the non-commutativity of quantum mechanical observables and density operators that makes a straightforward extension of the reasoning of the classical case difficult. Our approach that yields an exact quantum mechanical analogue of the classical bounds is based on the non-negativity of the relative entropy and additional trace inequalities for non-commuting self-adjoint  operators. In contrast to the classical framework, there are various different (and sensible) notions of relative entropy between the statistical distributions of quantum systems\cite{donald1986,hiai1991,moriya2020gibbs}; see also Eq.~(1.29) in Ref.~\citenum{carlen2019}. It turns out that the natural relative entropy analogue for our purposes is the von Neumann relative entropy that has been introduced by Umegaki \cite{umegaki1962}, since (a) it yields formally the same bounds as in the classical case and (b) it can be estimated by Monte Carlo methods. (We emphasize that there are different notions of divergences between probability measures in the classical case, too, beyond the relative entropy that is also known as \textit{Kullback-Leibler divergence}; cf.~Ref.~\citenum{hartmann2021nonasymptotic}).

The paper is organised as follows: We first review the results of Ref.~\citenum{jstatmecc} in Section \ref{sec:bogo2class} and then introduce the von Neumann relative entropy between the statistical operators associated with two quantum systems in Section \ref{sec:relentropyQuant}. Based on the properties of the von Neumann relative entropy, we derive bounds for the interface free energy in the quantum mechanical canonical ensemble in Section \ref{sec:bogo2quant}. The results are summarised and briefly discussed in Section \ref{sec:conclusions}.

\section{Two-sided Bogoliubov inequality for classical systems}\label{sec:bogo2class}
In this section we review the key concepts and results of Ref.~\citenum{jstatmecc} which are mandatory for the extension to the quantum case.
The interface energy resulting from the separation of a large system in independent subsystems, at positive temperature $T > 0$, is defined as the difference between the free energy of the system and the free energy of the uncoupled subsystems. Specifically, we consider a classical system bound to a volume $\Omega \subset \R^n$ which is described by a Hamiltonian $H : \Omega \to \R$. Assume that $H$ can be decomposed according to $H = H_0 + U$ where $H_0 := \sum_{i=1}^{d} H_i$, $H_i : \Omega_i \to \R$, $\Omega_i \subset \R^n$, $\bigcup_{i=1}^{d} \Omega_i = \Omega$, is the Hamiltonian describing the $d \in \N$ non-interacting subsystems and $U : \Omega \to \R$ is the coupling energy between those systems. It will be assumed that the functions $H_0$ and $U$ are continuous, sufficiently fast growing at infinity and bounded from below. The partition functions $Z$ and $Z_0$ associated to $H$ and $H_0$ are given by
\begin{align*}
    Z = \int_{\Omega} \ee^{-\beta H(x)} \, \diff^n x \tand Z_0 = \int_{\Omega} \ee^{-\beta H_0(x)} \, \diff^n x \ .
\end{align*}
We will refrain from indicating the dependency of the partition function on $\beta$, $\Omega$ and the particle numbers. In the following, we provide definitions and properties that are required for the final derivation of the upper and lower bounds of the interface energy.\\
\begin{definition}[Interface energy]
    The difference in free energy $\Delta F$ between the coupled system described by $H$ and the uncoupled subsystems described by $H_0$ is called \textbf{interface energy}, and it is given by
    \begin{align*}
        \Delta F := F - F_0 = - \beta^{-1} \log\left(\frac{Z}{Z_0}\right) \ .
    \end{align*}
\end{definition}

We briefly review the situation for a classical statistical ensemble. To this end, let $(\Omega, \Sigma, P)$ be a probability space (or: ensemble) where $\Sigma=\mathcal{B}(\Omega)$ denotes the $\sigma$-Algebra of Borel subsets of $\Omega$.  For convenience, we consider only probability measures with probability density function (pdf); specifically, we assume that there is an integrable, nonnegative function $p \colon \Omega \to [0,\infty)$ such that for all $A \in \Sigma$, it holds that
\begin{align*}
    P(A) = \int_{A} p(x) \, \diff^n x \ .
\end{align*}

\begin{definition}[Relative entropy]\label{def:partitioning_RE}
    Let $f, g : \Omega \to [0, \infty)$ be two pdfs on $\Omega$. Assume that
    \begin{align*}
        \int_{\{x \in \Omega \ : \ g(x) = 0\}} f(x) \, \diff^n x = 0 \ .
    \end{align*}
    The \textbf{relative entropy} (also known as Kullback-Leibler divergence $KL(f, g)$) between $f$ and $g$ is then defined as
    \begin{align*}
        R(f, g) := \int_{\Omega} \log\left(\frac{f(x)}{g(x)}\right) f(x) \, \diff^n x \ .
    \end{align*}
    In case that the integral of $f$ over the set of zeros of $g$ does not vanish (i.e. if $g \neq 0$ does not hold almost everywhere on $\Omega$ with respect to the probability measure $P$ induced by the density $f$), one defines $R(f, g) := \infty$. Note that the definition of $R$ is based on the limit $\lim_{x \to 0} x \log(x) = 0$.
\end{definition}


It is a simple consequence of Jensen's inequality that $R(f, g) \ge 0$,  with equality if and only if $f = g$ holds $P$-almost everywhere\cite{jstatmecc}.
\subsection{Two-sided Bogoliubov inequality}
In the following, it will always be assumed that the first argument of the relative entropy is strictly positive. This implies that null sets of the measure $P$ are Lebesgue null sets. Let us denote by $p$ and $p_0$ the pdfs of the canonical ensemble associated with the Hamiltonians $H$ and $H_0$, i.e.
\begin{align}\label{eq:partitioning_canonicalClassical}
    p := \frac{1}{Z} \, \ee^{-\beta H} \tand p_0 := \frac{1}{Z_0} \, \ee^{-\beta H_0} \ .
\end{align}
The expectation of any integrable random variable (or: observable) $O$ in the respective ensemble can then be written as
\begin{align*}
    \mathbf{E}_p[O] := \int_{\Omega} O(x) p(x) \, \diff^n x \tand \mathbf{E}_{p_0}[O] := \int_{\Omega} O(x) p_0(x) \, \diff^n x \ .
\end{align*}

\begin{theorem}[Two-sided Bogoliubov inequality]\label{thm:partitioning_bogoliubov}
    If the previous assumptions are satisfied, it follows that
    \begin{align}\label{eq:partitioning_bogoliubov}
        \mathbf{E}_p[U] \le \Delta F \le \mathbf{E}_{p_0}[U] \ .
    \end{align}
\end{theorem}

{For the full details of the proof we invite the reader to consult Ref.~\citenum{jstatmecc}, here we sketch it only in its essence:
\begin{proof}
    As $H_0$ and $U$ are continuous functions growing sufficiently fast at infinity, it follows that $U$ is integrable with respect to the densities $p$ and $p_0$, i.e. the expectation values are well-defined. Furthermore, $p$ and $p_0$ are strictly positive by construction. The non-negativity of the relative entropy implies:
    \begin{align*}
        0 &\le R(p, p_0) = \int_{\Omega} \log\left(\frac{p(x)}{p_0(x)}\right) p(x) \, \diff^n x = \int_{\Omega} \left[\log\left(\frac{\ee^{-\beta H(x)}}{\ee^{-\beta H_0(x)}}\right) + \log\left(\frac{Z_0}{Z}\right)\right] p(x) \, \diff^n x \\
        &= - \beta \int_{\Omega} \Bigl(H(x) - H_0(x)\Bigr) p(x) \, \diff^n x - \log\left(\frac{Z}{Z_0}\right) \int_{\Omega} p(x) \, \diff^n x \\
        &= - \beta \int_{\Omega} U(x) p(x) \, \diff^n x - \log\left(\frac{Z}{Z_0}\right) \ ,
    \end{align*}
    that is
    \begin{align*}
        \mathbf{E}_p[U] = \int_{\Omega} U(x) p(x) \, \diff^n x \le - \beta^{-1} \log\left(\frac{Z}{Z_0}\right) = \Delta F \ .
    \end{align*}
    Similarly, by interchanging the arguments of the relative entropy $R$, one obtains
    \begin{align*}
        0 &\le R(p_0, p) = \beta \int_{\Omega} U(x) p_0(x) \, \diff^n x + \log\left(\frac{Z}{Z_0}\right) = \beta \, \mathbf{E}_{p_0}[U] + \log\left(\frac{Z}{Z_0}\right) \ ,
    \end{align*}
    i.e. the desired upper bound
    \begin{align*}
        \Delta F \le \mathbf{E}_{p_0}[U] \ . \tag*{\qedhere}
    \end{align*}
\end{proof}}

Theorem \ref{thm:partitioning_bogoliubov} is a rigorous and powerful criterion to estimate the amount of statistical errors stemming from the microscopic nature of a system divided into non-interacting subsystems. It allows for a quantitative justification of the computation for a small system instead of a computationally unfeasible ideal system: simulating representative small subsystems in lieu of the fully  coupled system is justified if the interface energy $\Delta F$ is negligible compared to the energy scale of each subsystem. If the criterion holds, then  the complexity of the molecular simulation is reduced from, roughly, $\MO\bigl((3N)^{2}\bigr)$ to $\MO\bigl((3N_1)^{2} + \dotsb + (3N_d)^{2}\bigr)$ where $N_k$, $k = 1, \dotsc, d$, is the number of particles in the $k$-th subsystem and $N = N_1 + \dotsb + N_d$. If the criterion does not hold at a satisfactory level, then one has to revise the model of the system by modifying the interaction potential $U$ or by changing the size of the subsystems in order to incorporate effects resulting from the environment \cite{physrep}.

\begin{remark}
	The upper bound on $\Delta F$ is the well-known \textit{Bogoliubov inequality} or \textit{Peierls-Bogoliubov} or \textit{Gibbs-Bogoliubov inequality}\cite{peierls,liebleb,symanzik}.
	It is possible to improve the bounds using the Gibbs variational principle; specifically, for any integrable random variable $\phi$ and any positive pdf $f$, it holds that 
	\begin{equation}\label{eq:varbound}
		\mathbf{E}_{p}[U + \beta^{-1} \phi] - \beta^{-1} \log\bigl(\mathbf{E}_{p_0}[\ee^{\phi}]\bigr) \le \Delta F \le \mathbf{E}_{f}[U] + \beta^{-1} R(f, p_0) \ .
	\end{equation}  
\end{remark}

\section{Statistical operator and quantum relative entropy}\label{sec:relentropyQuant}
The bounds of Theorems \ref{thm:partitioning_bogoliubov} were proved in the framework of classical statistical mechanics, the natural question which arises at this point is whether an extension to quantum systems within the framework of quantum statistical mechanics is possible. This problem is addressed in this and in the following section. Specifically, in this section we will proceed with an elaboration of the framework of quantum statistics that allows the generalisation of the previous results. The key point in using the statistical or density operator (also known as density matrix) as the analogue of the phase-space pdf in the classical case are various trace inequalities that allow us to extend the concept of relative entropy to an equivalent quantum definition. The notion of relative entropy for quantum systems is not unique (e.g. Refs.~\citenum{donald1986,hiai1991}), and it turns out that the suitable concept for our purposes is the classical definition of Umegaki\cite{umegaki1962}, also termed the \textit{von Neumann relative entropy} in quantum information theory\cite{vedral}; cf.~also Ref.~\citenum{bookquantinf}. 

\subsection{Some density matrix theory}

We start by recapitulating the key concepts of statistical quantum mechanics, referring to the standard textbook of Zeidler\cite[Ch. 5.17]{Zeidler}. To begin with, we consider a quantum system on a complex separable Hilbert space $\MH$. Given an orthonormal basis $(\psi_n)_{n \in \N} \subset \MH$ and a sequence $(p_n)_{n \in \N} \subset \R$ of nonnegative real numbers with the properties
\begin{align*}
    \forall n \in \N \ : \ 0 \le p_n \le 1 \tand \sum_{n=1}^{\infty} p_n = 1 \ ,
\end{align*}
we refer to $\Psi := (\psi_n, p_n)_{n \in \N}$ as a statistical state of the system where $p_n$ is interpreted as the probability of finding the system in the state $\psi_n$. In the following, only \textit{mixed states}, characterised by $p_n < 1$ for all $n \in \N$, are of interest. If $T : \MH \supset \dom(T) \to \MH$ is a self-adjoint linear operator representing an observable, we define the expectation of $T$ in the statistical state $\Psi$ by 
\begin{align*}
    \mathbf{E}_{\Psi}[T] := \sum_{n=1}^{\infty} p_n \braket{\psi_n, T \psi_n} \ .
\end{align*}
Note that the expectation comprises \textit{statistical averaging} over the weights $p_n$ resulting from the statistical nature of the state $\Psi$ as well as \textit{quantum-mechanical averaging} $\braket{\psi_n, T \psi_n}$ resulting from the non-deterministic nature of quantum theory\cite[Ch. 2.1.1]{Nolting6}.

\begin{definition}[Statistical operator]
    A bounded self-adjoint linear operator $\rho : \MH \to \MH$ is called \textbf{statistical operator} if there are numbers $(p_n)_{n \in \N}  \subset [0, 1]$ with the property $\sum_{n=1}^{\infty} p_n = 1$ and an orthonormal basis $(\psi_n)_{n \in \N} \subset \MH$ such that the action of $\rho$ on $\psi \in \MH$ is given by
    \begin{align}\label{eq:partitioning_density}
        \rho \psi := \sum_{n=1}^{\infty} p_n \braket{\psi_n, \psi} \psi_n \ .
    \end{align}    
\end{definition}

Note that for all $m \in \N$, $\psi_m$ is an eigenfunction of $\rho$ with corresponding eigenvalue $p_m$, i.e. $\rho \psi_m = p_m \psi_m$. Furthermore, there is a one-to-one correspondence between statistical operators $\rho$ and statistical states $\Psi$ given by Eq. \eqref{eq:partitioning_density} (see \citet[Ch. 5.17, Prop. 2]{Zeidler}). Statistical operators are trace-class, hence compact. Therefore, they possess a discrete spectrum of eigenvalues with a corresponding orthonormal system of eigenvectors $(e_n)_{n \in \N} \subset \MH$ such that 
\begin{align*}
    \Tr(\rho) = \sum_{n=1}^{\infty} \braket{e_n, \rho e_n} = \sum_{n=1}^{\infty} \sum_{m=1}^{\infty} p_m \braket{\psi_m, e_n} \braket{e_n, \psi_m} = \sum_{m=1}^{\infty} p_m = 1 \ .
\end{align*}
Another relevant operator in this context, which is related to the concept of \textit{entropy}, is
\begin{align}\label{eq:partitioning_logdensity}
    \rho \log \rho : \MH \to \MH \quad , \quad (\rho \log \rho) \psi := \sum_{n=1}^{\infty} p_n \log(p_n) \braket{\psi_n, \psi} \psi_n \ .
\end{align}
(The mapping $\rho\mapsto \rho\log\rho$ is operator convex.) The expectation of an observable $T$ in a statistical state $\Psi$ can now be expressed in terms of the statistical operator $\rho$ as follows: 
\begin{align}\label{eq:partitioning_qmExp}
    \mathbf{E}_{\rho}[T] = \sum_{n=1}^{\infty} p_n \braket{\psi_n, T \psi_n} = \sum_{n=1}^{\infty} \braket{\rho \psi_n, T \psi_n} = \sum_{n=1}^{\infty} \braket{\psi_n, \rho T \psi_n} = \Tr(\rho T) \ .
\end{align}

Given a Hamiltonian $H : \MH \supset \dom(H) \to \MH$ with associated partition function $Z = \Tr(\ee^{-\beta H})$ which we assume to be finite, the canonical ensemble is described by the statistical operator
\begin{align*}
    \rho := \frac{ \ee^{-\beta H}}{\Tr(\ee^{-\beta H})} = \frac{\ee^{-\beta H}}{Z} \ .
\end{align*}
This is the quantum-mechanical generalisation of the pdf $p$ defined in Eq. \eqref{eq:partitioning_canonicalClassical}. 
A representation of $\rho$ of the form \eqref{eq:partitioning_density} is given in terms of the eigenfunctions $(\phi_n)_{n \in \N}$ of the Hamiltonian: letting $H \phi_n = E_n \phi_n$, we have $p_n = \ee^{- \beta E_n} / \sum_{n=1}^{\infty} \ee^{-\beta E_n} = \ee^{- \beta E_n} / Z$.

The two-sided Bogoliubov inequality is basically a consequence of the non-negativity of the quantum-mechanical relative entropy that we define next. {For a detailed overview of the history and the current mathematical status of the Bogoliubov inequality for quantum systems we refer the reader to the textbook by Zagrebnov on Gibbs semigroups \cite{zagrebnov}.}

\begin{definition}[Relative entropy\cite{umegaki1962}] 
    Let $\rho, \sigma : \MH \to \MH$ be two statistical operators. Define the quantum-mechanical \textbf{relative entropy} $\MR$ between $\rho$ and $\sigma$ to be
    \begin{align*}
        \MR(\rho, \sigma) := \Tr(\rho \log \rho) - \Tr(\rho \log \sigma) \ .
    \end{align*}
\end{definition}
The non-negativity of the relative entropy is a direct consequence of Klein's inequality that, for two positive trace-class operators $A,B: \MH\to\MH$ 
with $\Tr(A\log A)<\infty$, reads (see \citet[Lem.~14]{duan2017entropy})
\begin{equation}\label{eq:kleinsie}
\Tr(B\log B) \ge \Tr(B\log A + B - A)\ .
\end{equation}
(For an elementary proof in the finite-dimensional case, see \citet[App.~A]{carlen2019}.). {If one sets $B=\rho$ and $A=\sigma$, due to the fact that the statistical operators fulfill by definition the condition $\Tr \rho=1$ and $\Tr\sigma=1$, and under the assumption that $\Tr(\sigma\log\sigma) <\infty$ (trivially satisfied when $\MH$ is finite-dimensional, i.e. in applications of molecular simulation) one obtains that:}
\begin{equation}
  \MR(\rho, \sigma) = \Tr(\rho \log \rho) - \Tr(\rho \log \sigma) \ge 0 \ .
  \label{rineq}
\end{equation}
It also follows from the strict concavity of the logarithm that $\MR(\rho, \sigma)=0$ if and only if $\rho=\sigma$ in the sense that $p_n=q_n$ for all $n\in\N$, for which $p_n\neq 0$. 
{The crucial point of the current idea is that the non-negativity of $\MR$ can be used to derive the two-sided Bogoliubov inequality for statistical operators, as we will discuss next.}

\section{Two-sided quantum Bogoliubov inequality}\label{sec:bogo2quant}
Assume again that the Hamiltonian $H$ can be decomposed according to $H := H_0 + U$ and define 
\begin{equation}\label{eq:rho0rho}
\rho_0 := \frac{\ee^{-\beta H_0}}{Z_0}\,,\; Z_0 := \Tr(\ee^{-\beta H_0}) \quad \text{and}\quad  
\rho := \frac{\ee^{-\beta H}}{Z}\,,\; Z := \Tr(\ee^{-\beta H})\ .
\end{equation}
Using linearity of the trace and $\Tr(\rho) = 1$, we observe that 
\begin{align*}
    \MR(\rho_0,\rho) &= \Tr\left[\rho_0 \Bigl(\log(\ee^{-\beta H_0}) - \log(Z_0)\Bigr)\right] - \Tr\left[\rho_0 \Bigl(\log(\ee^{-\beta H}) - \log(Z)\Bigr)\right] \\
    &= - \beta \Tr(\rho_0 H_0) - \log(Z_0) + \beta \Tr(\rho_0 H) + \log(Z) \\
    &= \beta \Tr\Bigl(\rho_0 (H - H_0)\Bigr) + \log\left(\frac{Z}{Z_0}\right) \\
    &= \beta \Tr(\rho_0 U) + \log\left(\frac{Z}{Z_0}\right)
\end{align*}
which, together with inequality \eqref{rineq} implies the upper bound $\Delta F \le \mathbf{E}_{\rho_0}[U]$, with $\Delta F = -\beta^{-1}\log(Z/Z_0)$. This is the famous Peierls-Bogoliubov inequality\cite{peierls,liebleb,symanzik}; see also \citet[App.~A]{carlen2019}. 

The proof of the lower bound of the two-sided Bogoliubov inequality proceeds along the same line, using the reversed relative entropy $\MR(\rho,\rho_0)\ge 0$: 
\begin{align*}
    \MR(\rho,\rho_0) &= \Tr\left[\rho \Bigl(\log(\ee^{-\beta H}) - \log(Z)\Bigr)\right] - \Tr\left[\rho \Bigl(\log(\ee^{-\beta H_0}) - \log(Z_0)\Bigr)\right] \\
    &= \beta \Tr\Bigl(\rho (H_0 - H)\Bigr) + \log\left(\frac{Z_0}{Z}\right)\\
    &= - \beta \Tr(\rho U) - \log\left(\frac{Z}{Z_0}\right) \ .
\end{align*}
This entails the lower bound $\mathbf{E}_{\rho}[U] \le \Delta F$. {We summarize the calculation in the following theorem; an alternative proof, based on the differentiation of the exponential operator, can be found in \citet[Sec.~3.4; cf. Cor.~3.22 and Remark 3.23]{zagrebnov}.} 
\begin{theorem}\label{thm:partitioning_bogoliubov2}
	Let the partition functions $Z_0,Z>0$ in \eqref{eq:rho0rho} be finite and $U=H-H_0$, with $\mathbf{E}_{\rho_0}[U]<\infty$ and $\mathbf{E}_{\rho}[U]<\infty$. Then
	\begin{equation}\label{eq:partitioning_bogoliubov2}
        \mathbf{E}_{\rho}[U] \le \Delta F\le \mathbf{E}_{\rho_0}[U]\ .
	\end{equation}
\end{theorem}
Independently of the proof one may choose, an innovative aspect that needs to be underlined is that Theorem \ref{thm:partitioning_bogoliubov2} can actually be applied to molecular simulations of quantum systems and to define the error due to finite size approximations which are unavoidable in simulations (above all in simulations of quantum systems). In the language of simulations, the theorem prescribes the calculation of the average energy at the interface of the subsystems in which a large system of reference is divided. The calculation must be carried for: (a) when the subsystems interact through the standard particle-particle interactions (which essentially corresponds to the calculation of an ideal surface energy in the large system of reference) and (b) in case the subsystems are non-interacting (which corresponds to effectively running separate simulations of smaller sizes). {In Section \ref{sec:computational_protocol}, we will discuss the main features of a computational protocol for typical situations occurring in molecular simulations and specify explicitly the quantities involved, highlighting the practical utility of the result above.}

\subsection{Obtaining sharper bounds}

The inequalities can be sharpened by using the upper bound of the Peierls-Bogoliubov inequality and the Golden-Thompson trace inequality. Specifically, we have:
{\begin{lemma}[Gibbs variational principle]\label{lem:gibbsvari}
	Let $M_1(\MH)$ be the set of Hermitian positive trace-class operators on $\MH$ with unit trace (i.e. statistical operators). Further let $\sigma\in M_1(\MH)$ and $W$ be any self-adjoint positive operator on a suitable (dense) subspace of $\MH$, with compact resolvent, such that $-\beta^{-1}\log\sigma+W$ is a self-adjoint positive operator on the domain of $-\beta^{-1}\log\sigma$. Then, 
	\begin{equation}\label{eq:gibbsvari}
        -\beta^{-1}\log \Tr\!\left(\ee^{\log\sigma -\beta W}\right) = \inf_{\gamma\in M_1(\MH)} \left\{\Tr\!\left(\gamma W\right) + \beta^{-1}\MR(\gamma,\sigma)\right\} \ .
	\end{equation}
	If $\Tr(\sigma\log\sigma)<\infty$ and $\Tr(\ee^{-\beta W}W)<\infty$, the infimum is attained at  
	\begin{equation}
		\gamma^* = \frac{\ee^{\log\sigma -\beta W}}{\Tr\!\left(\ee^{\log\sigma-\beta W}\right)}\,.
	\end{equation}
\end{lemma}}
\begin{proof}
We consider the upper bound $\Delta F\le \mathbf{E}_{\rho_0}[U]$ in \eqref{eq:partitioning_bogoliubov2} where, without loss of generality, we may assume that $\Tr(\ee^{-\beta H_0})=1$, such that $Z_0=1$. 

The upper bound can then be recast as
\begin{equation}
	-\beta^{-1}\log \Tr\!\left(\ee^{\log\rho_0 - \beta U}\right) \le \Tr\!\left(\rho_0 U\right)\ .
\end{equation}
Introducing the new potential $V = U - \beta^{-1}\log\rho_0$ turns the last inequality into
\begin{equation}\label{eq:gibbs1}
	-\beta^{-1}\log \Tr\!\left(\ee^{- \beta V}\right) \le \Tr\!\left(\rho_0 V\right) + \beta^{-1}\Tr\!\left(\rho_0\log\rho_0\right)\ .
\end{equation}
{Note that $\rho_0$ is arbitrary, in that the inequality holds for any combination of density operators $\rho_0\in M_1(\MH)$ and semibounded observable $V$ on $\MH$;} therefore we write \eqref{eq:gibbs1} in what follows as 
\begin{equation}\label{eq:gibbs2}
	-\beta^{-1}\log \Tr\!\left(\ee^{- \beta V}\right) \le \Tr\!\left(\gamma V\right) + \beta^{-1}\Tr\!\left(\gamma\log\gamma\right)\,
\end{equation}
for any density operator $\gamma$, where equality is obtained either by setting $\gamma=\ee^{-\beta V}/\Tr(\ee^{-\beta V})$ for a given $V$, or by setting $V=-\beta^{-1}\log\gamma$ when $\gamma$ is given. 
{Shifting the potential by $V\mapsto V + \beta^{-1}\log\sigma=:W$ for some density operator $\sigma\in M_1(\MH)$ and applying the Golden-Thompson rule
\begin{equation}\label{eq:goldenthompson}
\Tr\!\left(\ee^{-(A+B)}\right) \le \Tr\!\left(\ee^{-A}\ee^{-B}\right)
\end{equation}
that holds for every pair $A,B$ of self-adjoint, positive operators on a suitable (dense) subspace of $\MH$, such that $B$ is relatively bounded by $A$, with $A$-bound for $B$ being less than 1 and both $\ee^{-A}$ and $\ee^{-B}$ being trace-class (see \citet[Thm.~1]{breitenecker1972note}; cf.~\citet[Thm.~4]{ruskai1972} or \citet[Thm.~2]{araki1973}), we obtain} 
\begin{equation*}
		-\beta^{-1}\log \Tr\!\left(\sigma \ee^{- \beta W}\right) \le 	-\beta^{-1}\log \Tr\!\left(\ee^{\log\sigma - \beta W}\right)  \le  \Tr\!\left(\gamma W\right) + \beta^{-1}\MR(\gamma,\sigma)\ .
\end{equation*}
{where we used Eq. \eqref{eq:gibbs2} in the second step.} Hence,
\begin{equation}\label{eq:gibbs3}
	-\beta^{-1}\log \Tr\!\left(\sigma \ee^{-\beta W}\right) \le \Tr\!\left(\gamma W\right) + \beta^{-1}\MR(\gamma,\sigma)\,.
\end{equation}
{To show that equality can be attained, note that the right-hand side is operator convex in $\gamma$ and consider a non-decreasing sequence $(W_n)_{n\in \N}$ of bounded self-adjoint operators with $W_n\to W$ in the sense of strong resolvent convergence.\cite{weidmann1997strong} 
Further assume  $\Tr(\sigma\log\sigma)<\infty$ and $\Tr(\ee^{-\beta W}W)<\infty$, and define the sequence $(\gamma_n)_{n\in\N}$ of density operators 
\begin{equation}
	\gamma_n = \frac{\ee^{\log\sigma -\beta W_n}}{\Tr\!\left(\ee^{\log\sigma-\beta W_n}\right)}\quad,\quad n\ge 1\,.
\end{equation}
By \citet[Prop. 10.1.13]{oliviera2009} and the boundedness assumption for $W_n$, this implies strong convergence $\|W_n h - Wh\|\to 0$ for any $h\in\MH$ and with $\|\cdot\|$ denoting the norm on $\MH$. Then $\MR(\gamma_n,\sigma)<\infty$ for all $n\ge 1$, and, by Fatou's Lemma (that entails lower semi-continuity of the trace), we have 
\begin{align*}
	\inf_{n\ge 1}\left\{\Tr\!\left(\gamma_n W\right) + \beta^{-1}\MR(\gamma_n,\sigma)\right\} & = \inf_{n\ge 1} \left\{\frac{\Tr\!\left(\ee^{\log\sigma -\beta W_n} (W-W_n)\right)}{\Tr\!\left(\ee^{\log\sigma-\beta W_n}\right)} -\beta^{-1}\log \Tr\!\left(\ee^{\log\sigma -\beta W_n}\right) \right\}\\
	& \le  -\beta^{-1}\log \liminf_{n\to\infty}\Tr\!\left(\ee^{\log\sigma -\beta W_n}\right)\\
	& \le -\beta^{-1}\log \Tr\!\left(\ee^{\log\sigma-\beta W}\right)\,.
\end{align*}
This, together with (\ref{eq:gibbs3}), shows that the infimum in (\ref{eq:gibbsvari}) is attained at
\begin{equation}
	\gamma^* = \frac{\ee^{\log\sigma -\beta W}}{\Tr\!\left(\ee^{\log\sigma-\beta W}\right)}\,.
\end{equation}
}
\end{proof}

{
\begin{remark}
	The relative boundedness assumption underlying the Golden-Thompson inequality (\ref{eq:goldenthompson}) basically states that the extra potential $W$ in (\ref{eq:gibbs3}) is a small perturbation of the Hamiltonian $-\beta^{-1}\log\sigma$ (e.g. for $\sigma\propto \exp(-\beta H_0)$) that preserves self-adjointness.  
\end{remark}}

The Gibbs variational principle \eqref{eq:gibbsvari} has a dual form, known by the name of \textit{Donsker-Varadhan variational principle} that expresses the relative entropy by a maximisation over observables: 

\begin{lemma}\label{lem:donskervaradhan}
	
Under the assumptions of Lemma \ref{lem:gibbsvari}, it holds for all $\gamma,\sigma\in M_1(\MH)$ with finite relative entropy $\MR(\gamma,\sigma)<\infty$ that  
{	\begin{equation}\label{eq:dv}
	\MR(\gamma,\sigma) = \sup_{{\theta\le 0}}\left\{\Tr(\gamma\theta) - \log \Tr\left(\ee^{\log\sigma+\theta}\right)\right\}\, 
	\end{equation}
where the supremum is over all self-adjoint negative operators defined on a dense subspace of $\MH$.}
\end{lemma}

\begin{proof}
	Setting $\theta=-\beta W$ in \eqref{eq:gibbs3} and noting that the resulting lower bound for $\MR(\gamma,\sigma)$ is operator concave in $\theta$ yields the desired statement. 
\end{proof}

We can now combine Lemmas \ref{lem:gibbsvari} and \ref{lem:donskervaradhan} with the inequality
\[
\Delta F\le\mathbf{E}_\gamma[U]+\beta^{-1}\MR(\gamma,\rho_0)\ ,
\]
that, by Lemma \ref{lem:gibbsvari}, holds for an arbitrary density matrix $\gamma$. Then by Theorem \ref{thm:partitioning_bogoliubov2}, we obtain after setting $\sigma=\rho_0$ and $W=U$ in \eqref{eq:gibbs3} and \eqref{eq:dv}:
\begin{cor}
Under the assumptions of Theorem \ref{thm:partitioning_bogoliubov2}, it holds 
\begin{equation}\label{eq:partitioning_bogoliubov3}
    \sup_{{V\ge 0}}\left\{\mathbf{E}_\rho[U -  V] - \beta^{-1}\log{\Tr\!\left(\ee^{\log\rho_0-\beta V}\right)} \right\} =  \Delta F =  \inf_{\gamma\in M_1(\MH)}\left\{\mathbf{E}_{\gamma}[U] + \beta^{-1}\MR(\gamma,\rho_0)\right\}\,.
\end{equation}
In particular, we have the family of two-sided bounds that is valid for any {positive} observable $V$ on $\MH$ and any density matrix $\gamma\in M_1(\MH)$: 
\begin{equation}\label{eq:partitioning_bogoliubov4}
	\mathbf{E}_{\rho}[U-V] - \beta^{-1}\log{\Tr\left(\ee^{\log\rho_0-\beta V}\right)} \le  \Delta F\le \mathbf{E}_{\gamma}[U] + \beta^{-1}\MR(\gamma,\rho_0)\ .
\end{equation}
{By the Golden-Thompson inequality, Eq. (\ref{eq:partitioning_bogoliubov4}) implies 
\begin{equation}\label{eq:partitioning_bogoliubov5}
	\mathbf{E}_{\rho}[U-V] - \beta^{-1}\log\mathbf{E}_{\rho_0}\!\left[\ee^{-\beta V}\right] \le  \Delta F\le \mathbf{E}_{\gamma}[U] + \beta^{-1}\MR(\gamma,\rho_0)
\end{equation}
where the lower bound can in general not be attained, unless $H_0$ and $V$ commute since in this case it holds that $\Tr(\ee^{\log\rho_0-\beta W}) = \mathbf{E}_{\rho_0}[\ee^{-\beta W}]$. If the operators do not commute, the left-hand side may be smaller than the right-hand side, so Eq. (\ref{eq:partitioning_bogoliubov5}) yields a slightly weaker lower bound. Finally, note that the bounds of Theorem \ref{thm:partitioning_bogoliubov2} are a special case obtained by setting $V=0$ and $\gamma=\rho_0$.}
\end{cor}

\section{Sketch of the computational protocol for molecular simulations}\label{sec:computational_protocol}
{A typical situation where an optimal criterion for the separation of a large system into smaller independent subsystems is of particular importance occurs in the determination of the optimal size of the simulation box for a molecular liquid at a given molecular density. In principle, one needs a large number of molecules so that at the electronic level, microscopic properties such as spectroscopic responses linked to, e.g., molecular bonding are well-described. However, the cost of a large simulation often goes beyond the available computational resources, and thus one needs to choose a system as small as possible while still being able to reasonably reproduce the properties of interest.

The criterion given in Theorem \ref{thm:partitioning_bogoliubov2} can be used to define the optimal size of the simulation box for electronic properties of a molecular liquid. Fig. \ref{cartoon} illustrates a typical setup of this kind for a static situation, and the example below treats the simple case of partitioning a large system into two smaller subsystems; the extension to several subsystems is straightforward.
\begin{figure}
  \centering\includegraphics[width=\figwidth]{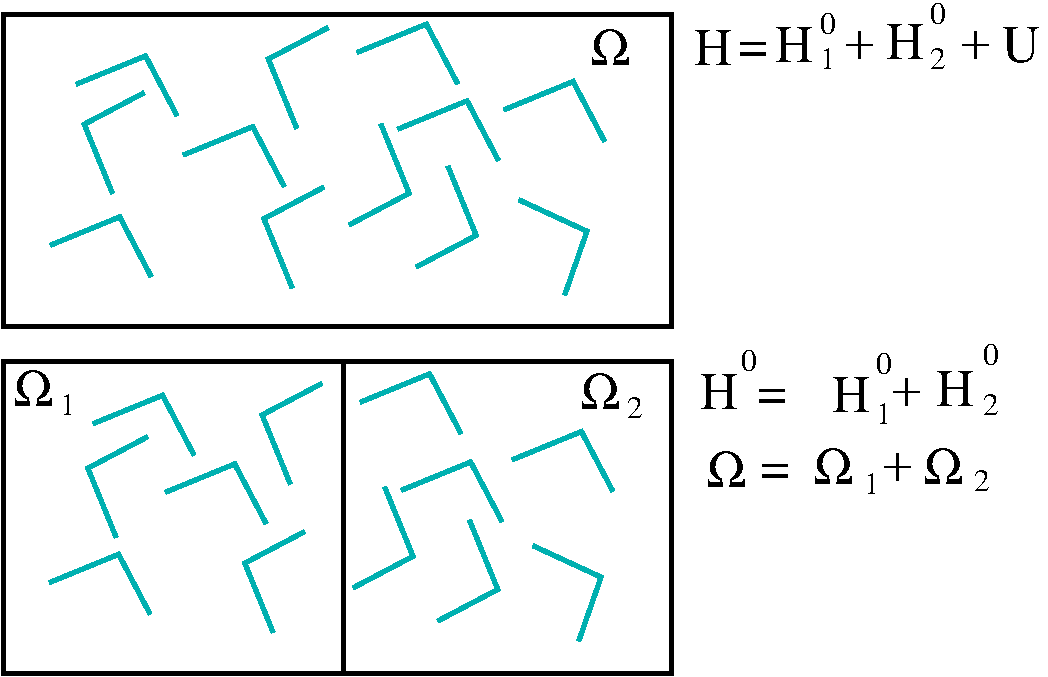}
  \caption{{The simulation box of a molecular liquid (e.g. water). The upper part schematically illustrates the large system with domain $\Omega$ and full Hamiltonian $H$ while the lower part illustrates the partitioning of the large system into two independent subsystems with domains $\Omega_{1}$ and $\Omega_{2}$ and corresponding Hamiltonians $H^{0}_{1}$ and $H^{0}_{2}$. The two subsystems can be treated in separate simulations.}}
  \label{cartoon}
\end{figure}
In the current example, the electronic Hamiltonian (in atomic units) of the whole system in the domain $\Omega$ takes the form
\begin{equation}\label{fullham}
    H=-\frac{1}{2}\sum_{i=1}^{N}\nabla_{i}^{2}+\sum_{1 \le i < j \le N}\frac{1}{|{\bf r}_{i}-{\bf r}_{j}|}-\sum_{I=1}^{M}\sum_{i=1}^{N}\frac{Z_{I}}{|{\bf R}_{I}-{\bf r}_{i}|}
\end{equation}
where $N$ is the total number of electrons, $M$ the total number of nuclei and  $Z_{I}$ the charge of the $I$-th nucleus; $\mathbf{r}_i$ and $\mathbf{R}_I$ denote the positions of the electrons and nuclei, respectively. The question is whether the approximation of considering only one of the subsystems, e.g. the one defined in the domain $\Omega_{1}$, would be sufficient to properly address the local electronic properties, and thus to avoid to include the rest of the box which occupies the domain $\Omega_{2}$.

The above question is equivalent to the problem of determining the degree of  independence of the two subsystems with respect to the larger system of reference as expressed by Theorem \ref{thm:partitioning_bogoliubov2}. In this context, the Hamiltonian for the system in the domain $\Omega_{1}$ reads
\begin{equation}
H^{0}_{1}=-\frac{1}{2}\sum_{i=1}^{n}\nabla_{i}^{2}+\sum_{1 \le i < j \le n}\frac{1}{|{\bf r}_{i}-{\bf r}_{j}|}-\sum_{I=1}^{W}\sum_{i=1}^{n}\frac{Z_{I}}{|{\bf R}_{I}-{\bf r}_{i}|}
\label{ham1}
\end{equation}
where $n$ is the number of electrons and $W$ is the number of nuclei in $\Omega_{1}$. Similarly, the Hamiltonian for the system in the domain $\Omega_{2}$ is given by
\begin{equation}
H^{0}_{2}=-\frac{1}{2}\sum_{i=1}^{m}\nabla_{i}^{2}+\sum_{1 \le i < j \le m}\frac{1}{|{\bf r}_{i}-{\bf r}_{j}|}-\sum_{I=1}^{Y}\sum_{i=1}^{m}\frac{Z_{I}}{|{\bf R}_{I}-{\bf r}_{i}|}
\label{ham2}
\end{equation}
where $m$ is the number of electrons and $Y$ is the number of nuclei in $\Omega_{2}$. Then the operator $U$ appearing in Theorem \ref{thm:partitioning_bogoliubov2} for the given instantaneous partitioning of the system takes the form
\begin{equation}
U=\sum_{i=1}^{n}\sum_{k=1}^{m}\frac{1}{|{\bf r}_{i}-{\bf r}_{k}|}-\sum_{K=1}^{Y}\sum_{i=1}^{n}\frac{Z_K}{|{\bf R}_{K}-{\bf r}_{i}|}-\sum_{I=1}^{W}\sum_{k=1}^{m}\frac{Z_I}{|{\bf R}_{I}-{\bf r}_{k}|}
\label{uform}
\end{equation}
with ${\bf r}_{i},{\bf R}_{I} \in \Omega_{1}$ for all $i=1, \dotsc, n$ and $I=1, \dotsc, W$, and ${\bf r}_{k},{\bf R}_{K} \in \Omega_{2}$ for all $k=1, \dotsc, m$ and $K=1, \dotsc, Y$.

With the partitioning of $\Omega$ defined above, the calculation of the quantities in Theorem \ref{thm:partitioning_bogoliubov2} can be done through the density matrix $\rho_{\Omega}$ of the full reference system corresponding to the Hamiltonian $H$ in the whole domain $\Omega$ while the density matrix corresponding to the two non-interacting subsystems $\rho_{1,2}=\rho_{\Omega_{1}}\otimes\rho_{\Omega_{2}}$ is determined by  $\rho_{\Omega_{1}}$ and $\rho_{\Omega_{2}}$ calculated in two separate simulations on $\Omega_{1}$ with $H^{0}_{1}$ and on $\Omega_{2}$ with $H^{0}_{2}$ respectively. Eq. \eqref{eq:partitioning_bogoliubov} then implies
\begin{equation}\label{eqex}
    \mathbf{E}_{\rho_{\Omega}}[U] \le \Delta F\le \mathbf{E}_{\rho_{1,2}}[U]\ .
\end{equation}
Furthermore, the definition of the operator $U$ shows that the interactions with respect to the electronic degrees of freedom are only of one-body and two-body form; thus, the density matrix representations needed for the calculations are one-body and two-body reduced density matrix terms, that is, three dimensional electron densities and the two-body electron-electron correlation (for example given by the electron radial distribution function $g(|{\bf r}-{\bf r'}|)$). Such quantities are routinely computed in electronic structure calculations and numerical schemes for quantum chemistry , e.g., Kohn-Sham Density Functional Theory \cite{tru}, Quantum Monte Carlo \cite{quantummonte} and high level quantum-chemical techniques \cite{quantchem}. In a dynamic simulation, the statistics is enlarged by repeating the procedure and considering several instantaneous, uncorrelated configurations along the molecular trajectory of the system.

If the mean coupling energies $\mathbf{E}_{\rho_{\Omega}}[U]$ and $\mathbf{E}_{\rho_{1,2}}[U]$ per molecule (i.e. divided by the number of atoms) have values of the order of the characteristic energy scale of the quantity of interest, such as the molecule-molecule energy bond per molecule, then one can conclude that the model error due to the chosen size of the (isolated) simulation box is too large. Conversely, when $\mathbf{E}_{\rho_{\Omega}}[U]$ and $\mathbf{E}_{\rho_{1,2}}[U]$ have values much smaller than the physical quantity of reference, one can reasonably trust in the simulation setup chosen. Once an optimal box size is determined, according to the protocol suggested here, this information can be used in all the successive simulations for the same quantities of interest. 
In perspective, one may be able to extend this idea to the definition of the optimal size of the quantum region in quantum mechanical/molecular mechanical simulation where a quantum region is embedded into a larger classical molecular systems \cite{qmmm}, or to the determination of the corrections required in the computational technique of molecular fragments where large polyatomic molecules such as polymers are divided into independent fragments and treated independently via quantum-chemical calculations \cite{fragm}, a technique which seems to be very promising for calculations on (futuristic) quantum computers \cite{qcomp}.}

\section{Discussion and Conclusions}\label{sec:conclusions}
Large systems of particles in fully atomistic resolution are a computational challenge for numerical simulations. The routinely used approximation is to treat small systems as representatives of large systems under the assumption that the influence of finite size effects is negligible in the computation of physical and chemical quantities of interest. The latter assessment requires precise and rigorous criteria of controlling these effects, otherwise modelling artefacts may easily deteriorate the predictions that can be obtained from finite systems. This work continues the efforts that were undertaken in a previous paper, in which a general criterion to precisely estimate the effect of finiteness of the system was developed for classical systems. Here, the extension to quantum systems is made by introducing the operator formalism for the equivalent classical quantities, namely, the statistical operator which is formally equivalent to the classical phase-space probability distribution and the von Neumann relative entropy that is commonly used in quantum information theory. 

In doing so, we have proved a two-sided Hilbert space version of the well-known Bogoliubov inequality that is applicable to simulation of infinite-dimensional quantum systems, regardless of whether these systems are fermionic or bosonic. The bounds can be useful in connection with electronic structure calculations for open systems where finite size effects are often a major burden in the development of efficient computational models \cite{trans,ceperl,advthsimquant}, or they can be used for bosonic and semiclassical systems in path integral molecular dynamics simulations, in which the control of finite size effects is the current bottleneck for applications in many fields of interest  \cite{parrbos,manolop,animesh,revtwo}. 

\acknowledgments{This work was partly supported by the DFG Collaborative Research Center 1114 ``Scaling Cascades in Complex Systems'', project No.235221301, Projects A05 (C.H.) ``Probing scales in equilibrated systems by optimal nonequilibrium forcing'' and C01 (L.D.S.) ``Adaptive coupling of scales in molecular dynamics and beyond to fluid dynamics''. {The authors thank Armen Allahverdyan for pointing out a flaw in the statement of the Gibbs variational principle in an earlier version of the manuscript.}\\ On behalf of all authors, the corresponding author states that there is no conflict of interest.\\ All data generated or analysed during this study are included in this published article.
}

\bibliography{literature}
\end{document}